\documentclass[11pt]{amsart}
\usepackage{amsmath, amsthm, amsfonts, amsbsy, amssymb, upref, enumerate, bigstrut,  color, mathtools, mathrsfs, float, bm}

\usepackage[left=1.3in,top=1.3in,right=1.3in]{geometry}

\usepackage{epsfig, graphicx}
\usepackage{ hyperref}

\newtheorem{theorem}{Theorem}[section]
\newtheorem{lemma}[theorem]{Lemma}

\newtheorem{proposition}[theorem]{Proposition}

\newcommand{\ra}{\rightarrow}

\newcommand{\smallavg}[1]{\langle #1 \rangle}

\newcommand{\nto}{\mbox{$\;\longrightarrow_{\hspace*{-0.42cm}{\small n}}\;$~}}

\newtheorem{thm1}{Theorem} 

\begin{document}
\title[Ultrametricity Property for the RFIM]
{{On the Ultrametricity Property in Random Field Ising Models
%in the Presence of Non-Gaussian Disorders 
%and Small External Field Perturbation
}}
\author{J. Roldan}
\address{
\newline
Departamento de Matem\'atica -
Universidade de Bras\'{i}lia, Brazil,
Email: \textup{\tt jinsupe10000@gmail.com}
}
%\thanks{Research partially supported by  CNPq}
%
%\author{Helton Saulo}
%\address{\newline Department of Statistics 
%\newline 
%Universidade de Bras\'{i}lia, Brazil
%\newline 
%\text{Helton Saulo}
%\newline
%Email: \textup{\tt heltonsaulo@gmail.com}}
%
\author{R. Vila}
\address{
\newline 
Departamento de Estat\'istica -
Universidade de Bras\'ilia, Brazil,
Email: \textup{\tt rovig161@gmail.com}
}
%\thanks{Research partially supported }

\date{\today}


\begin{abstract}
In this paper we show that the ultrametricity property remains valid in random field Ising models
with independent disorder 
whenever
the field strength is a small perturbation.
\end{abstract}

\maketitle
%\smallskip 
\section{Introduction}
\label{sec-1}
In statistical mechanics,
the random field Ising model (RFIM) \cite{imry1975random, Anatolyimry1970} 
is considered 
one of the simplest non-trivial models that belongs to a
class of disordered systems in which the disorder is
coupled to the order parameter of the system. 
This model is under intensive investigation 
and has been studied from several aspects. For example,
it is expected that many properties, as the Parisi ultrametricity 
(see \cite{Panchenkoultra, parisi1980sequence}) and the (extended) 
Ghirlanda-Guerra identities
(see \cite{aizenman1998stability, ghirlanda1998general,panchenko2007}), 
in disordered spin models 
should not depend on the particular distribution of the coupling constants.  
These properties are known to hold
in several mean-field spin glass models, such as the Sherrington-Kirkpatrick
model \cite{PhysRevLett.35.1792} and generic mixed $p$-spin models, 
see, e.g., \cite{Auffinger}.
The ultrametricity property was predicted by Parisi in \cite{parisi1980sequence} 
as an attempt to describe the expected behavior of the model and it still remains 
an unsolved mathematical problem.
On the other hand, in \cite{ghirlanda1998general} it was proven rigorously that
the Ghirlanda-Guerra identities hold (in the infinite volume limit)
in some approximate sense; for some specific choice of 
perturbed parameters \cite{Talagrand2010}. 
Results involving the ultrametricity property in spin glass models can be 
found, for example, in 
\cite{Auffinger, panchenko2010, panchenko2011ghirlanda, Panchenkoultra, book-talagrand}.

The main goal of this paper is to remove the hypothesis of Gaussian 
disorder and to show that the 
Parisi ultrametricity is
valid in RFIMs under mild assumptions on the disorder.
To this end, results similar to the ones in Chatterjee (2015) \cite{chatterjee2015disorder} 
are proven for the RFIM with independent disorder in the case that the field 
strength is a small perturbation. 
Afterward, we added to the random field  asymptotically vanishing 
non-Gaussian perturbations, and for this generalized model, 
we derived the extended Ghirlanda-Guerra identities, 
for the random models containing $p$-spin terms, for all $p\geqslant 1$. 
These identities combined with the main theorem of Panchenko (2011)
\cite{Panchenkoultra} establishes ultrametricity.

This paper is organized as follows. In Section \ref{sec:3},
we present the random models and state
our main result. Sections \ref{proof} and \ref{proof-2}, 
are dedicated to prove this paper's main result. 
In Appendix, we provide a proof for the integration by parts formula 
used in this paper. 

%%%%%%%%%%%%%%%%%%%%%%%%%%%%%%%%%%%%%%%%%%%%%%%%%%%%%%%%%%%%%%%%%%%%%

\section{Statement of the result}
\label{sec:3}	
Given $n\geqslant 1$, let 
$V_n=\mathbb{Z}^d\cap[1,n]^d$, $d\geqslant 1$, be a finite subset of vertices of 
$d$-dimensional hypercubic lattice 
with cardinality denoted by $|V_n|$.
The (random) Gibbs measure of the ferromagnetic RFIM on the set of spin 
configurations  
$\{\pm 1\}^{V_n}$  is given by
\begin{align}\label{gibbs-measure2}
{G}(\{\sigma\})
=
{1\over Z}
\exp\Bigg(
\beta \sum_{\langle xy\rangle}\sigma_x \sigma_y + (\mu-h)\sum_{x\in V_n} g_x \sigma_x
\Bigg)\,,
\end{align}
where $\langle xy\rangle$ denotes the set of ordered pairs in $V_n$ of 
nearest neighbors, $\beta>0$ and 
$\mu-h>0$ (with $\mu>0$),
called inverse temperature and field strength, respectively.
The partition function $Z$ appearing in the definition of 
${G}$ is a normalizing factor and 
$g_x$'s are
independent random variables 
(that collectively are called the disorder)
with 
zero-mean and unit-variance. Additionally, we assumed that  
the field strength is a small perturbation with the following decay rate, 
\begin{align}\label{condition-pert} 
h\nto \mu; \quad 
(\mu-h)\sqrt{|V_n|}\nto \infty; \quad 
{1\over |V_n|}\sum_{x\in V_n}
\mathbb{E}\big(|g_x|^{3}:|g_x|\geqslant \varepsilon (\mu-h)^{-1}\big)\nto 0\,,
\end{align}
for any $\varepsilon>0$. RFIMs with presence 
of disorder and field strength satisfying condition \eqref{condition-pert}
appeared in Auffinger and Chen (2016) \cite{Auffinger}. 
In the reference \cite{RV2019} the authors have studied 
the behaviour of RFIM with disorder having similar decay rate 
but different from ours. Additionally, the field strength used 
in \cite{RV2019} remains unchanged with respect to the volume.

Moreover, as in Panchenko (2013) \cite{panchenko2013sherrington} and Talagrand (2011) \cite{book-talagrand}, we add asymptotically vanishing perturbations to the Hamiltonian corresponding to Gibbs measure in \eqref{gibbs-measure2} to define the random Gibbs measure ${G}_{\boldsymbol{\alpha}}$ on $\{\pm 1\}^{V_n}$ with the following perturbing Hamiltonian
\begin{align}\label{gibbs-measure3}
\beta \sum_{\langle xy\rangle}\sigma_x \sigma_y + (\mu-h)\sum_{x\in V_n} g_x \sigma_x + H^{\rm per}_{n;\boldsymbol{\alpha}}(\sigma)
\,,
\end{align}
with
\begin{align}\label{def_Hper}
&H^{\rm per}_{n;\boldsymbol{\alpha}}(\sigma)
=
c_n\sum_{p\geqslant 2} \alpha_p 2^{-p}
H_{n;p}(\sigma);
\quad
H_{n;p}(\sigma)={1\over\vert V_n\vert^{(p-1)/2}}\sum_{x_1,\ldots,x_p} \xi_{x_1,\ldots,x_p} \sigma_{x_1}\cdots \sigma_{x_p}\,;
\end{align}
where the right-sided summation is over all $(x_1,\ldots,x_p)\in \otimes_{i=2}^p V_n$. 
Here, the sequence of numbers $(c_n)$ is such that $c_n\nto 0$, 
the sequence  $\boldsymbol{\alpha}=(\alpha_p)$ is given and satisfies 
$\vert\alpha_p\vert\leqslant 1$, 
and the disorder $(\xi_{x_1,\ldots,x_p})$ 
consists of i.i.d. real-valued random variables $\xi_{x_1,\ldots,x_p}$, for 
$p\geqslant 2$, with zero-mean and  unit-variance.
%and satisfies \eqref{condition-pert}, 
%and the following conditions: (a) the moment generating function (MGF) 
%of $\xi_{x_1,\ldots,x_p}$ exists and is finite; (b)
%$\Vert\xi\Vert_{\infty}\coloneqq\sup_{k,p\geqslant 2}\mathbb{E} \xi_{x_1,\ldots,x_p}^{2k}<\infty$.
%The disorders $(g_x)$ and $(\xi_{x_1,\ldots,x_p})$ are assumed to be independent of 
%each other. 
Note that when $\boldsymbol{\alpha}=\boldsymbol{0}$, ${G}_{\boldsymbol{\alpha}}(\sigma)={G}(\sigma)$ for all $\sigma\in\{\pm 1\}^{V_n}$.

For a function $f:(\{\pm 1\}^{V_n})^m\to\mathbb{R}$, $m\geqslant 1$, we define
\begin{align}\label{en-int}
\langle \,f\,\rangle_{\boldsymbol{\alpha}}
=
\langle f(\sigma^1,\ldots,\sigma^m)  \rangle_{\boldsymbol{\alpha}}
&=
\int f(\sigma^1,\ldots,\sigma^m)  \, \text{d}{G}_{\boldsymbol{\alpha}}(\sigma^1)\cdots \text{d}{G}_{\boldsymbol{\alpha}}(\sigma^m)
\,.
\end{align}
The randomness of the $g_x$'s and $\xi_{x_1,\ldots,x_p}$'s
will be represented by the measure 
$\gamma$ on $\mathbb{R}^{V_n}\times \mathbb{R}^{\otimes_{i\geqslant 2} V_n}$.
Following the notation of Talagrand (2003) \cite{talagrand2003spin}, we write 
\[
\nu_{\boldsymbol{\alpha}}(f)
\coloneqq 
\mathbb{E}\langle \,f\,\rangle_{\boldsymbol{\alpha}}
=
\int \langle \,f\,\rangle_{\boldsymbol{\alpha};g=u, \xi=v} \, \text{d}\gamma(u,v)
\,,
\]
averaging on the realizations of the disorder, where
$\langle\,\cdot\,\rangle_{\boldsymbol{\alpha};g=u,\xi=v}$ is 
the Gibbs expectation defined by setting 
$g_x$ and $\xi_{x_1,\ldots,x_p}$ in $\langle\,\cdot\,\rangle_{\boldsymbol{\alpha}}$ 
to be $u_x$ and $v_{x_1,\ldots,x_p}$ for each $p$, respectively, for each $x,x_i\in V_n$.
Since $\sqrt{\mathbb{E}H^2_{n;p}(\sigma)}\leqslant \sqrt{\vert V_n\vert}$, where $H_{n,p}$ is as in \eqref{def_Hper}, 
it follows from Lemma $3.6{''}$ of Loève (1951) \cite{Loeve-1951} that the series $c_n\sum_{p\geqslant 2} \alpha_p 2^{-p}
H_{n;p}(\sigma)$ converges almost surely. Therefore the Hamiltonian $c_n \alpha_p 2^{-p}
H_{n;p}(\sigma)$ is well-defined almost surely.

A collection $\sigma^1,\sigma^2,\ldots$ of configurations which 
are independent and identically distributed with respect to the 
Gibbs measure \eqref{gibbs-measure3} are known as replicas.
The spin overlap
between two replicas $\sigma^\ell$, $\sigma^{\ell'}$ is defined as
\begin{align}\label{overlap}
R_{\ell,\ell'}
\coloneqq 
{1\over |V_n|}\sum_{x\in V_n} 
\sigma^\ell_x \sigma^{\ell'}_x, \quad \forall \ell,\ell'\geqslant 1\,.
\end{align}
Note that $|R_{\ell,\ell'}|\leqslant 1$, $R_{\ell,\ell}=1$ 
and that the
infinite random array $R = (R_{\ell,\ell'})_{\ell,\ell'\geqslant 1}$ is symmetric, 
non-negative definite and weakly exchangeable.  
Following \cite{DSFinetti82}, an infinite random array $R$ with such properties 
is known as 
Gram-de Finetti matrix.
The array $R$ is said
to satisfy the extended Ghirlanda-Guerra identities 
(see \cite{aizenman1998stability,ghirlanda1998general, panchenko2007}) if 
for any $m\geqslant2$, any bounded measurable function $f=f\big((R_{\ell,\ell'})_{1\leqslant \ell,\ell'\leqslant m}\big)$, 
and for any bounded measurable function $\psi:\mathbb{R}\to\mathbb{R}$ of one overlap,
\begin{align}\label{G-G}
\nu_{\boldsymbol{\alpha}}(f \psi(R_{1,m+1}))
- 
\frac{1}{m}  \,
\nu_{\boldsymbol{\alpha}}(f)
\nu_{\boldsymbol{\alpha}}(\psi(R_{1,2}))
- 
\frac{1}{m}
\sum_{\ell=2}^m \nu_{\boldsymbol{\alpha}}(f \psi(R_{1,\ell}))
\nto 0\,.
\end{align}

In order to lighten the notation, we will omit the subscript $\boldsymbol{\alpha}$, 
when $\boldsymbol{\alpha}=\boldsymbol{0}$.

\bigskip

We are finally ready to state our main result.
\begin{thm1}\label{rsbthm}
	Under the assumption \eqref{condition-pert}, the following hold:	
	\begin{itemize}
		\item[\rm 1)] $R_{2,3}\geqslant \min\{R_{1,2},R_{1,3}\}$ almost surely 
		w.r.t. the infinite volume limit of $\nu$;
%		\vspace*{-0,2cm}
		\item[\rm 2)] $R_{2,3}\geqslant \min\{R_{1,2},R_{1,3}\}$ almost surely 
		w.r.t. the infinite volume limit of $\nu_{\boldsymbol{\alpha}}$, for all $p\geqslant 1$ in \eqref{def_Hper}.
	\end{itemize}
	Therefore for {\rm RFIMs}, defined by 
	\eqref{gibbs-measure2}-\eqref{condition-pert} and  
	\eqref{condition-pert}-\eqref{def_Hper}, the array $R$ is ultrametric.
\end{thm1}
The major step of the proof of Theorem \ref{rsbthm} requires 
a generalization of the Gaussian integration by parts, 
as in \cite{Auffinger, CARMONA2006215} and \cite{YTchen2019}.

Our first tool will be the following  proposition.
Its proof appears in Auffinger and Chen (2016) \cite{Auffinger}, Lemma 2.2. 
\begin{proposition}\label{prop-ap}
	Let $y$ be a random variable such
	that its first $k\geqslant 2$ moments match those of a Gaussian random variable. Suppose
	that $f\in C^{k+1}(\mathbb{R})$. For any $K\geqslant1$;
	\begin{align*}
	|\mathbb{E}yf(y)-\mathbb{E}f'(y)|
	\leqslant
	\dfrac{2(\|f^{(k-1)}\|_\infty + \|f^{(k)}\|_\infty)}{(k-1)!} \,
	\mathbb{E}(|y|^k:|y|\geqslant K)
	+
	\dfrac{(k+1)K}{k!} \, \|f^{(k)}\|_\infty \mathbb{E}|y|^k.
	\end{align*}
\end{proposition}
Our second tool will be the following  proposition. Its proof is presented in Appendix. 
This result is new and can be seen as a generalization of Proposition \ref{prop-ap} for the
bivariate case.
\begin{proposition}
	\label{prop-int-gen}
	Let $x,y$ be two independent random variables such
	that their first $k\geqslant 2$ moments match those of a Gaussian random variable. Suppose
	that $f\in C^{k+2}(\mathbb{R}^2)$. For any $K_1,K_2\geqslant1$;
	%	and $k\geqslant 2$;
	\begin{align*} 
	\biggl|&\mathbb{E}xyf(x,y)-\mathbb{E}{\partial^2 f(x,y)\over\partial x\partial y}\biggr|
	\\[0,2cm]
	&\leqslant
	{2\over (k-1)!}
	\left(
	{\biggl\|{\partial^{k-1} f\over \partial y^{k-1}}\biggr\|_{\infty}}
	+
	{\biggl\|{\partial^{k} f\over \partial y^{k}}\biggr\|_{\infty}}
	\right)
	\mathbb{E}(|x|:|x|\geqslant K_1) \,
	\mathbb{E}(|y|^k:|y|\geqslant K_2)
	\\[0,2cm]
	&\quad +
	{2\over (k-1)!}
	\biggl(
	{\biggl\|{\partial^{k} f\over \partial x^{k-1}\partial y}\biggr\|_{\infty}}
	+
	{\biggl\|{\partial^{k+1} f\over \partial x^{k}\partial y}\biggr\|_{\infty}}
	\biggr)
	\mathbb{E}(|x|^k:|x|\geqslant K_2)
	\\[0,2cm]
	&\quad +
	{2(k+1)K_1\over k!}
	\left(
	K_2
	{\biggl\|{\partial^{k} f\over \partial y^{k}}\biggr\|_{\infty}} 
	\mathbb{E}|y|^{k}
	+
	{\biggl\|{\partial^{k+1} f\over \partial x^{k}\partial y}\biggr\|_{\infty}}
	\mathbb{E}|x|^k
	\right) \nonumber
	\\[0,2cm]
	&\quad +
	{(k+1)K_1\over k!}
	\left(
	{\biggl\|{\partial^{k} f\over \partial y^{k}}\biggr\|_{\infty}} 
	\mathbb{E}(|y|^{k+1}: |y|\geqslant K_2)
	+
	{\biggl\|{\partial^{k+1} f\over \partial x^{k}\partial y}\biggr\|_{\infty}}
	\mathbb{E}|x|^k
	\right)
	\\[0,2cm]
	&\quad +
	{(k+1)\over k!} 
	\left(
	K_2
	{\biggl\|{\partial^{k} f\over \partial y^{k}}\biggr\|_{\infty}} 
	\mathbb{E}(|x|:|x|\geqslant K_1) \, 
	\mathbb{E}|y|^k
	+
	{\biggl\|{\partial^{k+1} f\over \partial x^{k}\partial y}\biggr\|_{\infty}}
	\mathbb{E}(|x|^{k+1}:|x|\geqslant K_1)
	\right).
	\end{align*}
\end{proposition}

\smallskip 

The rest of this paper is devoted to the proof of Theorem \ref{rsbthm}. 

%%%%%%%%%%%%%%%%%%%%%%%%%%%%%%%%%%%%%%%%%%%%%%%%%%%%%%%%%%%%%%%%%%%%%

\section{Proof of Item 1}\label{proof}
In this section we show the validity of the ultrametricity property by proving lack of replica symmetry breaking in the RFIM defined by \eqref{gibbs-measure2}-\eqref{condition-pert}. In other words, we  show
that the degree to which the values of a spin overlap differ from its expectation value vanishes
in the thermodynamic limit of the RFIM with weak disorder \eqref{condition-pert}.
Results involving absence of replica symmetry breaking in the RFIM and related models can be found, for example, in \cite{chatterjee2015disorder,Itoi2019,Itoi-2019,RV2019}.	

\smallskip 	
Using the same notation as in  Chatterjee (2015) \cite{chatterjee2015disorder}, 
for each $(\beta,\mu)\in (0,\infty)^2$, let us define
\begin{align}\label{free-energy-alpha-zero}
F_{n}\coloneqq \log Z_{n};
\quad 
\psi_{n} \coloneqq \dfrac{F_{n}}{|V_n|};
\quad 
p_{n}\coloneqq\mathbb{E}\psi_{n}\,,
\end{align}
where $\psi_{n}$ is known as the free-energy density.

The proof of the next result was inspired in  the proofs of Lemmas 
4.1 and 2.5 of references \cite{RV2019} and \cite{Itoi-2019}, respectively.
\begin{lemma}\label{Lemma-1}
	For each $(\beta,\mu)\in (0,\infty)^2$ and any $\varepsilon>0$ there is  $C_n(\varepsilon,\mu)=O(\sqrt{|V_n|})$ so that 
	\begin{align*}
	\mathrm{Var}(F_n)
	\leqslant 
	C_n|V_n|+o(\vert V_n\vert) 
	\,.
	\end{align*}
\end{lemma}
\begin{proof} We will use the same notation as in Lemma 4.1 of \cite{RV2019}. 
	Given two disorders $g=(g_x)$ and $g^{*}=(g_x^{*})$ consisting of independent random variables,
	for each $s\in[0,1]$, we define a new random field $G=(G_{x})$ as follows:
	$G_x\coloneqq \sqrt{s}\,g_x+\sqrt{1-s}\,g_x^{*}, \ x\in V_n.$
	Moreover, we also consider the following generating function
	\[
	\gamma_n(s)\coloneqq \mathbb{E}[\mathbb{E}^{*}F_n(G)]^2,
	\quad \text{with} \
	F_n(G)=\log Z_n(G)\,,
	\]
	where $\mathbb{E}$ and $\mathbb{E}^{*}$ denote expectation over $g$ and $g^*$, respectively.
	
	For each $s\in(0,1)$, a simple computation shows that
	\begin{align}\label{id-1} 
	{{\rm d}\gamma_n\over{\rm d} s}(s)
	&= 	
	\sum_x
	\mathbb{E}
	\left[
	{g_x\over\sqrt{s}}\,
	\mathbb{E}^{*} F_n(G) 
	\mathbb{E}^{*}
	{\partial F_n(G)\over \partial G_x}
	-
	\mathbb{E}^{*} F_n(G) 
	\mathbb{E}^{*} {g_x^{*}\over\sqrt{1-s}}
	{\partial F_n(G)\over \partial G_x}
	\right]\,.
	\end{align}
	
	Using the integration by parts (see Proposition \ref{prop-ap}),
	with $k=2$ and $K=\varepsilon (\mu-h)^{-1}$, for any $\varepsilon> 0$, 
	we obtain 
	\begin{align*}
	\biggl\vert
	\mathbb{E}
	\biggl[
	{g_x\over\sqrt{s}}\,
	\mathbb{E}^{*} F_n(G) 
	\mathbb{E}^{*}
	{\partial F_n(G)\over \partial G_x}
	\biggr]
	&-
	\mathbb{E}
	\left[
	{1\over\sqrt{s}}\,
	{\partial\over \partial g_x}
	\mathbb{E}^{*} F_n(G) 
	\,
	\mathbb{E}^{*}
	{\partial F_n(G)\over \partial G_x}
	\right]
	\biggr\vert
	\\[0,1cm]
	&\leqslant 
	{C_1\over\sqrt{s}} \, 
	\mathbb{E}(|g_x|^2:|g_x|\geqslant \varepsilon (\mu-h)^{-1})
	+
	{C_2\over\sqrt{s}} \varepsilon (\mu-h)^{-1}
	\end{align*}
	and
	\begin{align*}
	\biggl\vert
	{\mathbb{E}^{*}F_n(G)} 
	\mathbb{E}^{*} {g_x^{*}\over\sqrt{1-s}}
	{\partial F_n(G)\over \partial G_x}
	&-
	{\mathbb{E}^{*}F_n(G)\over\sqrt{1-s}} 
	\mathbb{E}^{*}
	{\partial \over \partial g_x^{*}}
	{\partial F_n(G)\over \partial G_x}
	\biggr\vert
	\\[0,1cm]
	&\leqslant 
	{C_3\over\sqrt{1-s}} \, 
	\mathbb{E}(|g_x|^2:|g_x|\geqslant \varepsilon (\mu-h)^{-1})
	+
	{C_4\over\sqrt{1-s}} \varepsilon (\mu-h)^{-1}\,,
	\end{align*}
	where $C_1,C_2,C_3$ and $C_4$ are positive constants.
	
	By combining the last two inequalities with \eqref{id-1}, we get the following estimate
	\begin{align}\label{ineq-1} \footnotesize
	&\biggl\vert{{\rm d}\gamma_n\over{\rm d} s}(s)  
	-
	\sum_x
	\mathbb{E}
	\left[
	{1\over\sqrt{s}}\,
	{\partial\over \partial g_x}
	\mathbb{E}^{*} F_n(G) 
	\,
	\mathbb{E}^{*}
	{\partial F_n(G)\over \partial G_x}
	-
	{\mathbb{E}^{*}F_n(G)\over\sqrt{1-s}} 
	\,
	\mathbb{E}^{*}
	{\partial \over \partial g_x^{*}}
	{\partial F_n(G)\over \partial G_x}
	\right] \biggl\vert  \nonumber
	\\[0,1cm]
	&\leqslant 
	\biggl({C_1\over\sqrt{s}}+{C_3\over\sqrt{1-s}}\biggr) 
	\sum_{x}\mathbb{E}(|g_x|^2:|g_x|\geqslant \varepsilon (\mu-h)^{-1})
	+
	\biggl({C_2\over\sqrt{s}}+{C_4\over\sqrt{1-s}}\!\biggr)|V_n| \varepsilon (\mu-h)^{-1}
	\nonumber
	\\[0,1cm]
	&\eqqcolon
	\Theta(n,s,\varepsilon,\mu) \,.
	\end{align}
	
	On other hand, by using the identities
	\[
	{\partial F_n(G)\over \partial g_x} 
	=
	\sqrt{s}\, {\partial F_n(G)\over \partial G_x}; 
	\quad
	{\partial F_n(G)\over \partial g_x^{*}} 
	=
	\sqrt{1-s}\, {\partial F_n(G)\over \partial G_x},
	\]
	in \eqref{ineq-1}, for each $s\in(0,1)$, we have
	\begin{align*}
	\biggl\vert
	{{\rm d}\gamma_n\over{\rm d} s}(s)
	-
	\sum_x
	\mathbb{E}
	\Big(
	\mathbb{E}^{*}{\partial F_n(G)\over \partial G_x}
	\Big)^2
	\biggl\vert
	\leqslant 
	\Theta(n,s,\varepsilon,\mu)\,.
	\end{align*}
	Since $	\mathbb{E}
	(\mathbb{E}^{*}{\partial F_n(G)\over \partial G_x})^2
	\leqslant \mu^2$, it follows from the above inequality that
	\begin{align*} \footnotesize %\fontsize{0.05}{0.05}\selectfont
	{{\rm d}\gamma_n\over{\rm d} s}(s)
	\leqslant
	{\mu^2}|V_n|
	+
	\Theta(n,s,\varepsilon,\mu)\,.
	\end{align*}
	Combining the relation 
	$\mathrm{Var}(F_n)
	=\gamma_n(1)-\gamma_n(0)=\int_{0}^{1}{ {\rm d}\gamma_n\over{\rm d} s}(s)\, \text{d}s$
	with the above inequality we obtain
	\begin{align*}
	&\mathrm{Var}(F_n)
	\leqslant
	\mu^2|V_n|+\int_{0}^{1}\Theta(n,s,\varepsilon,\mu)\, \text{d}s
	\\
	&=	
	\big[\mu^2+ (C_2+C_4) \varepsilon (\mu-h)^{-1}\big]|V_n|
	+
	(C_1+C_3)
	\sum_{x}\mathbb{E}(|g_x|^2:|g_x|\geqslant \varepsilon (\mu-h)^{-1})
	\\
	&\eqqcolon C_n(\varepsilon,\mu)|V_n|+ o(|V_n|),
	\end{align*}
	thus completing the proof.
\end{proof}

For any $n\geqslant 1$, let
\begin{align}\label{delta-n}
\Delta_n
\coloneqq {1\over |V_n|} \sum_{x\in V_n} g_x \sigma_x
\end{align}
be the part of the energy function \eqref{gibbs-measure2} due to the disorder. 
Let ${\mathcal A}$ be the countable set of all $(\beta,\mu)\in (0,\infty)^2$
such that 
$\textstyle{\partial p\over \partial \mu^-}(\beta,\mu)\neq {\partial p\over \partial \mu^+}(\beta,\mu).$
\begin{proposition}\label{cond-main} 
	For any $(\beta,\mu)\in {\mathcal A}^c$, we have
	\begin{align*}
	\nu(\Delta_n)
	\nto
	{\partial p\over\partial \mu}(\beta,\mu);
	\quad 
	\mathbb{E}|\smallavg{\Delta_n}-\nu(\Delta_n)| 
	\nto 0\, .
	\end{align*}  
\end{proposition}
\begin{proof}
	Note that the same result was proved in Lemma 2.7 of Chatterjee (2015) \cite{chatterjee2015disorder} regardless of distribution of the disorder. Three ingredients are fundamental to prove this result: first, the convexity of $\psi_n$.  Second, the variance of $F_n$ does not grow faster than $\sqrt{\vert V_n\vert}$, which is guaranteed by Proposition \ref{Lemma-1}. Third, the limit $p= \lim_{n\ra\infty} p_{n}$ exists and is differentiable at $\mu$, 
	which is guaranteed by Lemmas 2.1 and 2.7 in \cite{chatterjee2015disorder}. 
	Therefore, the proof follows.
\end{proof}

The Proposition \ref{cond-main} plays an important role in the proof of the next result.
\begin{lemma}\label{lemma-main}
	Under the hypothesis of Theorem  \ref{rsbthm},	
	for any $(\beta,\mu)\in {\mathcal A}^c$, we have
	\[
	\nu\big(|\Delta_n - \nu(\Delta_n) | \big)  \nto 0\, .
	\]
\end{lemma}
\begin{proof}
	Let 
	$\smallavg{\sigma_x;\sigma_y}\coloneqq\smallavg{\sigma_x \sigma_y}-
	\smallavg{\sigma_x}\smallavg{\sigma_y}$ be the truncated two-point correlation.
	A straightforward computation shows that
	\begin{align*}
	\biggl|{\partial \smallavg{\sigma_x;\sigma_y}\over \partial g_y}\biggr|\leqslant 2\mu;
	\quad 
	\biggl|{\partial^2 \smallavg{\sigma_x;\sigma_y}\over \partial g_x \partial g_y}\biggr|\leqslant 6\mu^2;
	\quad 
	\biggl|{\partial^3 \smallavg{\sigma_x;\sigma_y}\over \partial g_x^2 \partial g_y}\biggr|\leqslant 24\mu^3\,.
	\end{align*}
	Let $\langle\,\cdot\,\rangle_{g_x=u, g_y=v}$ be the Gibbs expectation defined by setting 
	$g_x$ and $g_y$ in $\langle\,\cdot\,\rangle$ to be $u$ and $v$ respectively, and
	$F_{x,y}(u,v)\coloneqq \smallavg{\sigma_x;\sigma_y}_{g_x=u,g_y=v}$.
	Integrating by parts (see Proposition \ref{prop-int-gen})
	with  $f_{x,y}(u,v)=\mathbb{E}F_{x,y}(u,v)$, $k=2$ and $K_1=K_2=\varepsilon (\mu-h)^{-1}$, for any $\varepsilon> 0$, gives 
	\begin{align*}
	\biggl|\mathbb{E}g_xg_y f_{x,y}
	&-
	\mathbb{E}{\partial^2 f_{x,y}(g_x,g_y)\over\partial u\partial v}\biggr|
	\\
	&\leqslant
	4\mu(1+3\mu)\,
	\mathbb{E}(|g_x|:|g_x|\geqslant \varepsilon (\mu-h)^{-1}) \,
	\mathbb{E}(|g_y|^2:|g_y|\geqslant \varepsilon (\mu-h)^{-1})
	\\
	&\quad +
	12\mu^2 (1+4\mu)\,
	\mathbb{E}(|g_x|^2:|g_x|\geqslant \varepsilon (\mu-h)^{-1})
	+ 
	18\varepsilon(\varepsilon+4\mu^2)
	\\
	&\quad +
	9\varepsilon\mu
	\big[
	\mathbb{E}(|g_y|^{3}: |g_y|\geqslant \varepsilon (\mu-h)^{-1})
	+
	4\mu
	\big]
	\\
	&\quad 
	+
	9\mu(\varepsilon+4\mu^2)\,
	\mathbb{E}(|g_x|^{3}:|g_x|\geqslant \varepsilon (\mu-h)^{-1})
	\,.
	\end{align*}
	Dividing this inequality by $|V_n|^2$ and summing over all $x,y\in V_n$,
	the triangle inequality gives 
	\begin{align*}
	\mathbb{E}\big(\smallavg{\Delta_n^2}-\smallavg{\Delta_n}^2\big)
	&\leqslant
	{1\over |V_n|^2}
	\sum_{x,y}
	\biggl|\mathbb{E}g_xg_y f_{x,y}
	-
	\mathbb{E}{\partial^2 f_{x,y}(g_x,g_y)\over\partial u\partial v}\biggr|
	\\
	&\leqslant
	4\mu(1+3\mu)
	{1\over |V_n|^2}
	\biggl[\,
	\sum_x
	\mathbb{E}(|g_x|^2:|g_x|\geqslant \varepsilon (\mu-h)^{-1})\,
	\biggr]^2
	\\
	&\quad +
	12\mu^2
	(1+4\mu)\, 
	{1\over |V_n|}
	\sum_x
	\mathbb{E}(|g_x|^2:|g_x|\geqslant \varepsilon (\mu-h)^{-1})
	+
	18\varepsilon(\varepsilon+4\mu^2)
	\\
	&\quad 
	+
	9\varepsilon \mu\,
	{1\over |V_n|}
	\sum_y\mathbb{E}(|g_y|^{3}: |g_y|\geqslant \varepsilon (\mu-h)^{-1})+36\varepsilon \mu^2
	\\
	&\quad +
	9\mu(\varepsilon+4\mu^2)
	{1\over |V_n|}
	\sum_x
	\mathbb{E}(|g_x|^{3}:|g_x|\geqslant \varepsilon (\mu-h)^{-1})
	\,.
	\end{align*}
	Combining the last estimate with the following inequality
	\[
	\nu(|\Delta_n-\langle \Delta_n\rangle|)
	\leqslant
	\sqrt{\mathbb{E}(\langle \Delta_n^2\rangle-\langle \Delta_n\rangle^2)}\,,
	\]
	it follows from \eqref{condition-pert}, and the fact that $\varepsilon$ is arbitrary, 
	that
	$
	\nu\big(|\Delta_n-\langle \Delta_n\rangle|\big)\nto 0.
	$
	Finally, the proof follows from triangle inequality and Proposition \ref{cond-main}.
\end{proof}
\begin{proposition}
	Under the hypothesis of Theorem  \ref{rsbthm},	for any $(\beta,\mu)\in {\mathcal A}^c$, the following ergodic property holds:
	\begin{align*}
	\nu
	\Bigg(\, \biggl|
	{\beta\over|V_n|}
	\sum_{\langle xy\rangle}\big(\sigma_x\sigma_y-\nu(\sigma_x\sigma_y)\big)
	\biggr|\, \Bigg)
	\nto 0\,.
	\end{align*}
\end{proposition}
\begin{proof}	
	Let 
	\begin{align*}
	H_n(\sigma)=\beta \sum_{\langle xy\rangle}\sigma_x \sigma_y + h\sum_{x} g_x \sigma_x, \quad \sigma\in \{\pm 1\}^{V_n}\,,
	\end{align*} 
	be the Hamiltonian corresponding to the Gibbs measure \eqref{gibbs-measure2} of our RFIM.
	It is well-known that 
	\begin{align}\label{convergence}
	\nu\Bigg(\, \left|{H_n\over |V_n|}-
	\nu\biggl( {H_n\over|V_n|} \biggr) \right|\,\Bigg)  \nto 0\,;
	\end{align}
	see  \cite{PANCHENKO2010189, panchenko2013sherrington}. 
	Therefore follows from the second triangular inequality that 
	\begin{multline*}
	\nu\Bigg(\, \left|{H_n\over |V_n|}-
	\nu\biggl( {H_n\over|V_n|} \biggr) \right|\,\Bigg)
	\geqslant \
	\Biggl|\,
	\nu
	\Bigg(\, \biggl|
	{\beta\over|V_n|}
	\sum_{\langle xy\rangle}\big(\sigma_x\sigma_y-\nu(\sigma_x\sigma_y)\big)
	\biggr|\, \Bigg)
	-
	(\mu-h)\,
	\nu\big( | \Delta_n-
	\nu( \Delta_n) |\big)\,
	\Biggl|\,.
	\end{multline*}
	By combining this inequality with \eqref{convergence}, \eqref{condition-pert}, and 
	Lemma \ref{lemma-main}, the proof of the ergodic property follows.
\end{proof}

Our next step is to prove the Ghirlanda-Guerra identities for our model, which 
is precisely stated in the next lemma. In order to obtain these identities the Lemma \ref{lemma-main} plays 
an important role.

\begin{lemma}
	\label{ggid}
	Given $m\geqslant 2$, let $f:\mathbb{R}^{m(m-1)/2}\ra[-1,1]$ be a bounded measurable function
	of the overlaps \eqref{overlap} that not change with $n$.
	Then, under assumption \eqref{condition-pert}, the Ghirlanda-Guerra identities with $\psi\equiv{\rm Id}$ and $\boldsymbol{\alpha}=\boldsymbol{0}$ in \eqref{G-G} are satisfied at almost all $(\beta,\mu)$. 
	That is, if $f$ is as above, 
	\begin{align*} 
	\nu(f R_{1,m+1})
	- 
	\frac{1}{m}  \,
	\nu(f)
	\nu(R_{1,2})
	- 
	\frac{1}{m}\,
	\sum_{s=2}^m \nu(f R_{1,s})
	\nto 0\,,
	\quad \forall (\beta,h)\in {\mathcal A}^c\,.
	\end{align*}
\end{lemma}
\begin{proof}
	Since $\|f\|_\infty\leqslant 1$, we have 
	\begin{align}\label{first-part}
	\left|\nu\big(\Delta_{n}(\sigma^1) f\big)-\nu\big(\Delta_{n}(\sigma^1)\big)\nu\big(f\big)\right|
	\leqslant 
	\nu\big(|\Delta_{n} - \nu(\Delta_{n}) | \big)\,,
	\end{align}
	where $\Delta_{n}$ is as in \eqref{delta-n}.
	
	On the other hand, let $\langle\,\cdot\,\rangle_{g_x=u}$ be the Gibbs expectation defined by setting 
	$g_x$ in $\langle\,\cdot\,\rangle$ to be $u$ and
	$F_x(u)\coloneqq \langle\,\sigma_x^1 \, f\,\rangle_{g_x=u}$.
	%	Note that $F_x$ have bounded continuous third-order derivatives.
	Using \eqref{en-int},
	a  straightforward computation shows that 
	\begin{align}\label{derivative-main}
	{\partial^j F_x(u)\over \partial u^j}
	=
	(\mu-h)^j
	\biggl\langle\sigma_x^{1} \boldsymbol{\cdot} \Big(\sum_{\ell=1}^m\sigma_{x}^\ell-m\sigma_x^{m+1}\Big)^j f \biggr\rangle_{g_x=u},
	\quad j=1,2,\ldots\,.
	\end{align}
	We remark that the above formula has also appeared in Chen (2019) \cite{YTchen2019}.
	
	Since $|{\partial^j F_x(u)\over \partial u^j}|\leqslant(2m\mu)^j$ and $\mathbb{E}g_x^2=1$, 
	the integration by parts formula (Proposition \ref{prop-ap})
	with $f_x(u)\coloneqq \mathbb{E}F_x(u)$, $k=2$ and 
	$K=\varepsilon (\mu-h)^{-1}$, gives
	\begin{align*}
	\biggl|\mathbb{E}g_x f_x-\mathbb{E}{\text{d} f_x(g_x)\over \text{d} u}\biggr|
	\leqslant
	4m\mu(1+2m\mu)  \, \mathbb{E}(|g_x|^2:|g_x|\geqslant \varepsilon (\mu-h)^{-1})
	+
	6\varepsilon m^2\mu,
	\end{align*}
	for any $\varepsilon> 0$. 
	Dividing the above inequality by $|V_n|$, summing over $x\in V_n$, 
	applying the triangle inequality, and using \eqref{derivative-main} with $j=1$ 
	we get that  
	\begin{align*}
	\biggl|
	\nu\big(\Delta_n(\sigma^1) f\big)
	&- 
	(\mu-h)
	\nu\biggl(\Big(\sum_{\ell=1}^{m}R_{1,\ell}-mR_{1,m+1}\Big)f \biggr) 
	\biggr|
	\\
	&
	\leqslant
	{1\over|V_n|}\sum_{x} 
	\biggl|\mathbb{E}g_x f_x-\mathbb{E}{\text{d} f_x(g_x)\over \text{d} u}\biggr|
	\\[0,1cm]
	&\leqslant
	4m\mu(1+2m\mu)  \, 
	{1\over|V_n|}\sum_{x}\mathbb{E}(|g_x|^2:|g_x|\geqslant \varepsilon (\mu-h)^{-1})
	+
	6\varepsilon m^2\mu.
	\end{align*}
	Therefore, from both the assumption \eqref{condition-pert} and arbitrariness of $\varepsilon$, it follows that	
	\begin{align}\label{approx}
	\limsup_{n\to\infty}
	\sup_{f}
	\left|
	\nu\big(\Delta_n(\sigma^1) f\big)
	-
	(\mu-h)\nu\biggl(\Big(\sum_{\ell=1}^{m}R_{1,\ell}-mR_{1,m+1}\Big)f \biggr)	
	\right|=0\,.
	\end{align}
	In the particular case where $f=1$ and $m=1$, we obtain that 
	$
	|\nu(\Delta_{n}(\sigma^1))-(\mu-h)\nu(R_{1,1}-R_{1,2})|\nto 0.
	$
	By combining this with \eqref{first-part} and \eqref{approx}, it follows from Lemma \ref{lemma-main} that
	\begin{align*}
	\biggr|
	\nu(R_{1,1}-R_{1,2}) \nu\big(f\big)
	-
	\nu\biggl(\Big(\sum_{\ell=1}^{m}R_{1,\ell}-mR_{1,m+1}\Big)f \biggr)
	\biggr|
	\nto 0.
	\end{align*}
\end{proof}

The last ingredient of the proof of Item (1) of Theorem \ref{rsbthm} 
is the following result by Auffinger and Chen (2016) \cite{Auffinger}
which proves the self-averaging of the spin overlap.

\begin{proposition}[\cite{Auffinger}]
	%[Self-averaging of the overlap]
	\label{mainlmm-il}
	Under assumption \eqref{condition-pert},
	for any $(\beta,\mu)\in (0,\infty)^2$,
	\begin{align*}
	\mathbb{E}\left(\smallavg{R_{1,2}^2}
	- 
	\smallavg{R_{1,2}}^2 
	\right)
	=
	\nu(R_{1,2}-\smallavg{R_{1,2}})^2
	\nto 0;
	\quad 
	\nu\big(m(\sigma)-\smallavg{m(\sigma)}\,\big)^2\nto 0\,,
	\end{align*}
	where $m(\sigma)= \sum_x\sigma_x/|V_n|$ define the magnetization.
\end{proposition}

\vspace*{0,3cm}
\noindent
{\it Proof of Item 1 of Theorem \ref{rsbthm}.}   
At this point, we already have established the validity of the key ingredients: the self-averaging of the overlap (Proposition \ref{mainlmm-il})  and the Ghirlanda-Guerra identities (Lemma \ref{ggid}). 
By a quite standard argument (see, e.g., Chatterjee (2015) \cite{chatterjee2015disorder}) 
these two mentioned facts imply the replica symmetry breaking does not occur in the RFIM. 
In other words, the spin overlap is concentrated at its expectation value 
and then the ultrametricity property follows.
%For the convenience of the reader, we present this proof here.
%
%Choosing $m=2$ and $f = R_{1,2}$ in Lemma \ref{ggid} gives 
%%
%\begin{equation}\label{gg1}
%\nu(R_{1,2}R_{1,3})
%- 
%\frac{1}{2}  
%\nu(R_{1,2})^2
%- 
%\frac{1}{2}
%\nu(R_{1,2}^2)
%\nto 0\, . 
%\end{equation}
%%
%Choosing $m=3$ and $f= R_{2,3}$ gives
%%
%\begin{align}\label{gg2}
%\nu(R_{2,3} R_{1,4}) 
%- 
%\frac{1}{3}  
%\nu(R_{1,2})^2
%- 
%\frac{1}{3}
%\sum_{\ell=2}^{3}
%\nu(R_{2,3}R_{1,\ell})
%\nto 0\, . 
%\end{align}
%%
%By symmetry between replicas,
%%
%$
%\nu(R_{2,3}R_{1,2})
%= 
%\nu(R_{2,3}R_{1,3})
%= 
%\nu(R_{1,2}R_{1,3}).
%$
%Then, multiplying \eqref{gg1} by ${2/ 3}$ and adding to \eqref{gg2} we obtain
%%
%\begin{align}\label{limit-final}
%{2\over 3}
%\big[
%\nu(R_{1,2}^2)
%-
%\nu(R_{1,2})^2
%\big]
%-
%\mathbb{E}\big(\smallavg{R_{1,2}^2}
%- 
%\smallavg{R_{2,3}R_{1,4}} 
%\big)
%\nto 0\, .
%\end{align}
%%
%Since the sequence $(\sigma^{l})$ is independent,
%$
%\smallavg{R_{2,3}R_{1,4}} 
%%= 
%%{1\over |V_n|^2}\sum_{x,y\in V_n}h_x^2 h_y^2 
%%\, \langle\sigma^2_x \sigma^3_x \sigma^1_y \sigma^4_y\rangle
%%\\[0,1cm]
%%&=
%%{1\over |V_n|^2}\sum_{x,y\in V_n}h_x^2 h_y^2 
%%\, \langle\sigma_x\rangle^2 \langle\sigma_y\rangle^2
%=
%\smallavg{R_{1,2}}^2\,.
%$
%By combining this with \eqref{limit-final}, from
%Proposition \ref{mainlmm-il} the proof of Theorem \ref{rsbthm} follows.
\qed

%%%%%%%%%%%%%%%%%%%%%%%%%%%%%%%%%%%%%%%%%%%%%%%%%%%%%%%%%%%%%%%%%%%%%

\section{Proof of Item 2}\label{proof-2}

To prove this item, which is the ultrametricity property, 
%we look the RFIM   \eqref{condition-pert}-\eqref{def_Hper}  as a mixed $p$-spin model in the disorder $\widetilde{\xi}=(g_x, {\xi}_{x_1,\ldots,x_p})$  for $p\geqslant 2$, where the disorders $(g_x)$ and $({\xi}_{x_1,\ldots,x_p})$ are defined in \eqref{condition-pert} and in  \eqref{def_Hper}, respectively. Since both disorders are assumed to be independent of each other, working only on the randomness of the weak disorder $(g_x)$ and averaging on  the realizations of the disorder  $({\xi}_{x_1,\ldots,x_p})$, notice that all results from the previous section can easily be adapted to the mixed $p$-spin model  \eqref{condition-pert}-\eqref{def_Hper}. Then 
the strategy is to combine the results of the previous section with the following known 
results in the literature, concerning to mixed $p$-spin models for  $p\geqslant 2$:
\begin{itemize}
%	\vspace*{-0,2cm}
	\item  the main theorem of Panchenko (2010) \cite{PANCHENKO2010189}; 
	%Let $p\geqslant 2$ and let $Z_{\boldsymbol{\alpha}}(t)$ be the partition function defined by the Hamiltonian $c_{n}\alpha_p  2^{-p}
	%H_{n;p}(\sigma)$ with $\alpha_p=t.$ Suppose that
	%\begin{align}
	%& {1\over \vert V_n\vert}\,
	%\mathbb{E}
	%\big\vert\log Z_{\boldsymbol{\alpha}}(t)- \mathbb{E}\log Z_{\boldsymbol{\alpha}}(t) 
	%\big\vert\nto 0,
	%\\[0,1cm]
	%& {1\over \vert V_n\vert}\,
	%\mathbb{E}\log Z_{\boldsymbol{\alpha}}(t) \nto P(t) \ \text{in some neighborhood of} \ t, 
	%\\[0,1cm]
	%& P(t) \ \text{is differentiable at} \ t. 
	%\end{align}
	%Then
%	\vspace*{-0,2cm}
	\item the universality of Ghirlanda-Guerra identities in mixed $p$-spin models; see Chen (2019) \cite{YTchen2019};
%	\vspace*{-0,2cm}
	\item and the main theorem of Panchenko (2011) \cite{Panchenkoultra}.
%	\vspace*{-0,2cm}
\end{itemize}

\smallskip 
For any $n\geqslant 1$, let $\widetilde{\xi}_{x_1,\ldots,x_p} \coloneqq g_x\delta_{p,1}+{\xi}_{x_1,\ldots,x_p}(1-\delta_{p,1})$, where $\delta$ is the Kronecker delta function, and 
\begin{align*}
\Delta_{n;p}
\coloneqq{1\over\vert V_n\vert^{(p+1)/2}}\sum_{x_1,\ldots,x_p} \widetilde{\xi}_{x_1,\ldots,x_p} 
\sigma_{x_1}\cdots \sigma_{x_p}, \quad p\geqslant 1,
\end{align*}
the part of the energy function \eqref{gibbs-measure3} due to the disorder.
When $p=1$, $\Delta_{n;p}$ coincides with the random function $\Delta_{n}$ given in \eqref{delta-n}.

Next, we prove a very important and technical result used 
to obtain the extended Ghirlanda-Guerra identities, which is the following limit 
\begin{align*}%\label{conv}
%\lim_{n\to\infty}
\nu_{\boldsymbol{\alpha}}\big(|\Delta_{n;p} - \nu_{\boldsymbol{\alpha}}(\Delta_{n;p}) | \big)\nto 0.
\end{align*}

This result is established in Panchenko (2010) \cite{PANCHENKO2010189} and Auffinger and Chen (2018) \cite{Auffinger-other}.  
Actually, in \cite{Auffinger-other} a stronger result is obtained but at the cost of 
requiring a stronger hypothesis than the one used in \cite{PANCHENKO2010189}.
To be more precise, in \cite{PANCHENKO2010189} Panchenko 
obtain the above limit under the following assumptions: 
\begin{itemize}
	\item[a)] 
	$
	%	\lim_{n\to\infty} 
	\vert V_n\vert^{-1}\,
	\mathbb{E}
	\big\vert\log Z_{\boldsymbol{\alpha}}(t)- \mathbb{E}\log Z_{\boldsymbol{\alpha}}(t) 
	\big\vert\nto 0,$
	\item[b)]  
	$
	%	\lim_{n\to\infty} 
	\vert V_n\vert^{-1}\,
	\mathbb{E}\log Z_{\boldsymbol{\alpha}}(t) \nto P(t)$ \ \text{in some neighborhood of} \ $t$, and
	\item[c)] $P(t)$ \ \text{is differentiable at} \ $t$. 
\end{itemize}
On the other hand, this limit is obtained in Auffinger and Chen's work \cite{Auffinger-other} 
under the following assumptions:
\begin{itemize}
	\item[i)] there exists a nonrandom
	function $P: I=(t - \varepsilon, t + \varepsilon) \to \mathbb{R}$, for some $\varepsilon>0$, such that for any $t'\in I$, $\vert V_n\vert^{-1}\log Z_{\boldsymbol{\alpha}}(t')\nto P(t')$ almost surely, and 
	\item[ii)]  $P(t)$ is differentiable at $t$,
\end{itemize}

We can show that conditions i) and ii) implies a), b) and c).
But, as observed before, the conclusion in \cite{Auffinger-other} is
stronger than the one in \cite{PANCHENKO2010189}. 
However, in our setting is more convenient to establish the validity of 
conditions a), b) and c). 

\begin{proposition}\label{remain-proof} For any $p\geqslant 1$ 
	\begin{align}\label{convergence-fin}
	\nu_{\boldsymbol{\alpha}}\big(|\Delta_{n;p} - \nu_{\boldsymbol{\alpha}}(\Delta_{n;p}) | \big)\nto 0\,.
	\end{align}
\end{proposition}
\begin{proof}
	For almost all $(\beta,\mu)$,
	the convergence in \eqref{convergence-fin} is proved in Lemma \ref{lemma-main}
	% (with the respective modifications) 
	for $p = 1$. 
	For $p\geqslant 2$, let $Z_{\boldsymbol{\alpha}}(t)$ be the partition function 
	associated to the Hamiltonian $c_{n}\alpha_p  2^{-p}
	H_{n;p}(\sigma)$ with $\alpha_p=t$.
	Under the condition of two matching 
	moments of the disorder $({\xi}_{x_1,\ldots,x_p})$, 
	Lemma 8 of Carmona and Hu (2006) \cite{CARMONA2006215} 
	shows that the limit of the free-energy function
	$\vert V_n\vert^{-1}\log Z_{\boldsymbol{\alpha}}(t)$ 
	does not depend on the particular distribution of environment. 
	Therefore this limit is everywhere differentiable in $t$, for fixed $t$. 
	By using the martingale difference argument by Burkholder’s inequality \cite{Burkholder73}
	%, as in the proof of Lemma 8 in \cite{CARMONA2006215}, 
	and integration by parts formula, Chen (2019) \cite{YTchen2019} has proved that
	\begin{align*}
	\vert V_n\vert^{-2}
	\mathbb{E}
	\big\vert
	\log Z_{\boldsymbol{\alpha}}(t)-\mathbb{E}\log Z_{\boldsymbol{\alpha}}(t) 
	\big\vert^2\nto 0.
	\end{align*}
	Therefore, the hypotheses a), b), and c) of the main theorem in \cite{PANCHENKO2010189}
	are satisfied for the mixed $p$-spin model with $p\geqslant 2$. 
	Consequently,  \eqref{convergence-fin} is also valid for all $p\geqslant 2$. 
	
\end{proof}
%
%A stronger assumption than Panchenko’s for to prove the above proposition can be found in Auffinger and Chen (2018) \cite{Auffinger-other}.
%\begin{remark}
%{\color{red}
%Another alternative way to obtain Proposition \ref{remain-proof} is to check the following hypotheses of the main theorem of Auffinger and Chen (2018) \cite{Auffinger-other}: 
%\begin{itemize}
%	\item[i)] there exists a 
%	function $P: I=(t - \varepsilon, t + \varepsilon) \to \mathbb{R}$, for some $\varepsilon>0$, such that for any $t'\in I$, $\lim_{n\to\infty}\vert V_n\vert^{-1}\log Z_{\boldsymbol{\alpha}}(t')= P(t')$ almost surely, and
%	\vspace*{-0,2cm} 
%	\item[ii)]  $P(t)$ is differentiable at $t$,
%\end{itemize}
%Notice that 
%the hypotheses i) and ii) of Auffinger and Chen (2018) \cite{Auffinger-other} are stronger than Panchenko’s (2010) \cite{PANCHENKO2010189} ones a), b) and c).
%}
%\end{remark}

\vspace*{0,3cm}
\noindent
{\it Proof of Item 2 of Theorem \ref{rsbthm}.} 
In order to derive the extended Ghirlanda-Guerra identities \eqref{G-G}, 
we can assume, without loss of generality, that $\vert\psi\vert\leqslant 1$.
Since the space of real-valued compactly supported continuous functions on $[-1, 1]$ is dense in the space of Lebesgue integrable functions on $[-1, 1]$ and 
since any continuous function can be uniformly approximated on $[-1, 1]$ by a polynomial, it is sufficient to prove the extended
Ghirlanda-Guerra identities \eqref{G-G} for any higher moments of the overlap, that is, 
for all $p\geqslant 1$,
\begin{align}\label{G-G-2-3}
\nu_{\boldsymbol{\alpha}}(f R_{1,m+1}^p)
- 
\frac{1}{m}  \,
\nu_{\boldsymbol{\alpha}}(f)
\nu_{\boldsymbol{\alpha}}(R_{1,2}^p)
- 
\frac{1}{m}
\sum_{\ell=2}^m \nu_{\boldsymbol{\alpha}}(f R_{1,\ell}^p)
\nto 0, \quad m\geqslant2\,,
\end{align}
for any bounded measurable function $f=f\big((R_{\ell,\ell'})_{1\leqslant \ell,\ell'\leqslant m}\big):\mathbb{R}^{m(m-1)/2}\ra[-1,1]$. 

Indeed, firstly we consider the case $p=1$. In this case, the validity of identities \eqref{G-G-2-3}, 
follows from Proposition \ref{remain-proof} and Lemma \ref{ggid}.
Secondly, the validity of \eqref{G-G-2-3} for $p\geqslant 2$, is consequence of Proposition \ref{remain-proof} 
and the integration by parts formula (as in the proof of Theorem 2.1-Step 2, in Chen (2019) \cite{YTchen2019}).
Therefore, in the infinite volume limit, the Gibbs measure of the RFIM defined in  
\eqref{condition-pert}-\eqref{def_Hper}  
satisfies the extended Ghirlanda-Guerra identities \eqref{G-G} 
for any $p\geqslant 1$.

Finally, since $R= (R_{\ell,\ell'})_{\ell,\ell'\geqslant 1}$ is a Gram-de Finetti matrix \cite{DSFinetti82} and since the extended Ghirlanda-Guerra identities \eqref{G-G} are satisfied, 
%in the mixed $p$-spin model  ($p\geqslant 1$) \eqref{ham-p-spin}, 
the main theorem of Panchenko (2011) \cite{Panchenkoultra} establishes the ultrametricity.
\qed

%%%%%%%%%%%%%%%%%%%%%%%%%%%%%%%%%%%%%%%%%%%%%%%%%%%%%%%%%%%%%%%%%%%%%

\section*{Acknowledgements}
We would like to thank S. Chatterjee for raising the question on 
the ultrametricity property in the RFIM in general non-Gaussian disorders. 
We would like to offer our special thank to
L. Cioletti for many valuable comments and
careful reading of this manuscript.
This study was financed in part by the Coordena\c{c}\~{a}o de 
Aperfei\c{c}oamento de Pessoal de N\'{i}vel Superior
- Brasil (CAPES) - Finance Code 001.
Jamer Roldan was supported by CNPq.

%%%%%%%%%%%%%%%%%%%%%%%%%%%%%%%%%%%%%%%%%%%%%%%%%%%%%%%%%%%%%%%%%%%%%
%\appendix
%\numberwithin{equation}{section}
%\makeatletter 
%\newcommand{\section@cntformat}{
%	%\thesection:\ 
%}
%\makeatother
%\noindent
\section*{Appendix}
\begin{proof}[Proof of Proposition \ref{prop-int-gen}]
	Let $g(x,y)\coloneqq xf(x,y)$. 
	By Taylor expanding the function $g$ at $(k-1)$-th and $k$-th orders, we have
	\begin{align}
	yg(x,y)
	&=
	g(x,0)y
	-
	{\partial^{k-1} g(x,0)\over\partial y^{k-1}} \, {y^k\over (k-1)!}
	+
	\sum_{n=1}^{k-1} {\partial^n g(x,0)\over\partial y^n} \, {y^{n+1}\over n!}   \label{A1}
	\\[0,2cm]
	& \quad +
	{\partial^{k-1} g(x,a(y))\over\partial y^{k-1}} \, {y^k\over (k-1)!} \,, \nonumber   
	\\[0,2cm]
	&=
	g(x,0)y
	+
	\sum_{n=0}^{k-1}{\partial^n g(x,0)\over\partial y^n} \, {y^{n+1}\over n!}
	+
	{\partial^{k} g(x,b(y))\over\partial y^{k}} \, {y^{k+1}\over k!}.    \label{A2}
	\end{align}
	Similarly, considering the Taylor expansion of ${\partial g(x,y)\over\partial y}$ 
	up to its $(k-1)$-th order, we get 
	\begin{align}
	{\partial g(x,y)\over\partial y}
	&=
	\sum_{n=1}^{k-1} {\partial^n g(x,0)\over\partial y^n} \, {y^{n-1}\over (n-1)!}
	+
	{\partial^k g(x,c(y))\over\partial y^k} \, {y^{k-1}\over (k-1)!}\,,   \label{A3}
	\end{align}
	where $a(y), b(y), c(y)$ are some functions depending only on $y$.
	From \eqref{A1} and \eqref{A3} it follows that 
	\begin{align}\label{A}
	yg(x,y)-{\partial g(x,y)\over\partial y}
	&=
	g(x,0)y
	-
	{\partial^{k-1} g(x,0)\over\partial y^{k-1}} \, {y^k\over (k-1)!}
	+
	\sum_{n=1}^{k-1} {\partial^n g(x,0)\over\partial y^n} \, 
	h_n(y)
	\\[0,2cm]
	&\quad +
	{\partial^{k-1} g(x,a(y))\over\partial y^{k-1}} \, {y^k\over (k-1)!}
	-
	{\partial^k g(x,c(y))\over\partial y^k} \, {y^{k-1}\over (k-1)!}\,,    \nonumber
	\end{align}
	where $h_n(y)\coloneqq({y^{n+1}\over n!}-{y^{n-1}\over (n-1)!})$.
	From \eqref{A2} and \eqref{A3}, we get 
	\begin{align}\label{B}
	yg(x,y)-{\partial g(x,y)\over\partial y}
	&=
	g(x,0)y
	+
	\sum_{n=1}^{k-1} {\partial^n g(x,0)\over\partial y^n} \, 
	h_n(y)
	\\[0,2cm]
	&\quad +
	{\partial^{k} g(x,b(y))\over\partial y^{k}} \, {y^{k+1}\over k!} 
	-
	{\partial^k g(x,c(y))\over\partial y^k} \, {y^{k-1}\over (k-1)!}\,.   \nonumber
	\end{align}
	
	Expanding  ${\partial f(x,y)\over \partial y}$ 
	up to $(k-1)$-th and $k$-th orders, we have 
	\begin{align}
	x{\partial f(x,y)\over \partial y}
	&=
	{\partial f(0,y)\over \partial y} x
	-
	{\partial^{k} f(0,y)\over\partial x^{k-1}\partial y} \, {x^k\over (k-1)!}
	+
	\sum_{n=1}^{k-1} {\partial^{n+1} f(0,y)\over\partial x^n\partial y} \, {x^{n+1}\over n!} \label{A4}
	\\[0,2cm]
	&\quad +
	{\partial^{k} f(\widetilde{a}(x),y)\over\partial x^{k-1}\partial y} \, {x^k\over (k-1)!}\,,  
	\nonumber 
	\\[0,2cm]
	&=
	{\partial f(0,y)\over \partial y} x
	+
	\sum_{n=1}^{k-1} {\partial^{n+1} f(0,y)\over\partial x^n\partial y} \, {x^{n+1}\over n!}
	+
	{\partial^{k+1} f(\widetilde{b}(x),y)\over\partial x^{k}\partial y} \, {x^{k+1}\over k!}.
	\label{A5}
	\end{align}
	Again, expanding ${\partial^2 f(x,y)\over\partial x\partial y}$ up to
	$(k-1)$-th orders, we obtain
	\begin{align}
	{\partial^2 f(x,y)\over\partial x\partial y}
	&\!=\!
	\sum_{n=1}^{k-1} {\partial^{n+1} f(0,y)\over\partial x^n\partial y} \, {x^{n-1}\over (n-1)!}
	+
	{\partial^{k+1} f(\widetilde{c}(x),y)\over\partial x^k\partial y} \, {x^{k-1}\over (k-1)!} \,,
	\label{A6}
	\end{align}
	where $\widetilde{a}(x), \widetilde{b}(x), \widetilde{c}(x)$ are some functions 
	depending only on $x$.
	From \eqref{A4} and \eqref{A6},
	\begin{align}\label{C}
	x{\partial f(x,y)\over \partial y}
	-
	{\partial^2 f(x,y)\over\partial x\partial y}
	&\!=\!
	{\partial f(0,y)\over \partial y} x
	-\!
	{\partial^{k} f(0,y)\over\partial x^{k-1}\partial y} \, {x^k\over (k-1)!}
	\!+\!
	\sum_{n=1}^{k-1} {\partial^{n+1} f(0,y)\over\partial x^n\partial y} \, 
	h_n(x)
	\\[0,2cm]
	&\quad
	+
	{\partial^{k} f(\widetilde{a}(x),y)\over\partial x^{k-1}\partial y} \, {x^k\over (k-1)!}
	-
	{\partial^{k+1} f(\widetilde{c}(x),y)\over\partial x^k\partial y} \, {x^{k-1}\over (k-1)!} \,.
	\nonumber
	\end{align}
	From \eqref{A5} and \eqref{A6} we have
	\begin{align}\label{D}
	x{\partial f(x,y)\over \partial y}
	-
	{\partial^2 f(x,y)\over\partial x\partial y}
	&=
	{\partial f(0,y)\over \partial y} x
	+
	\sum_{n=1}^{k-1} {\partial^{n+1} f(0,y)\over\partial x^n\partial y} \, 
	h_n(x)
	\\[0,2cm]
	&\quad +
	{\partial^{k+1} f(\widetilde{b}(x),y)\over\partial x^{k}\partial y} \, {x^{k+1}\over k!}
	-
	{\partial^{k+1} f(\widetilde{c}(x),y)\over\partial x^k\partial y} \, {x^{k-1}\over (k-1)!}\,.
	\nonumber
	\end{align}
	Summing \eqref{A} and \eqref{C} one get that
	\begin{align}\label{E}
	xyf(x,y)
	-
	{\partial^2 f(x,y)\over\partial x\partial y}
	&=
	f(x,0)xy
	+
	{\partial f(0,y)\over \partial y} x
	-
	{\partial^{k-1} f(x,0)\over\partial y^{k-1}} \, {xy^k\over (k-1)!}
	\\[0,2cm]
	&\quad  -
	{\partial^{k} f(0,y)\over\partial x^{k-1}\partial y} \, {x^k\over (k-1)!}   \nonumber  
	\\[0,2cm]
	&\quad 
	+
	\sum_{n=1}^{k-1}
	\left(
	{\partial^n f(x,0)\over\partial y^n} \, x
	h_n(y)
	+
	{\partial^{n+1} f(0,y)\over\partial x^n\partial y} \, 
	h_n(x)
	\right)     \nonumber
	\\[0,2cm]
	&\quad +
	{\partial^{k-1} f(x,a(y))\over\partial y^{k-1}} \, {xy^k\over (k-1)!}
	-
	{\partial^k f(x,c(y))\over\partial y^k} \, {xy^{k-1}\over (k-1)!}   \nonumber
	\\[0,2cm]
	&\quad
	+
	{\partial^{k} f(\widetilde{a}(x),y)\over\partial x^{k-1}\partial y} \, {x^k\over (k-1)!}
	-
	{\partial^{k+1} f(\widetilde{c}(x),y)\over\partial x^k\partial y} \, {x^{k-1}\over (k-1)!}\,. \nonumber
	\end{align}
	By adding up \eqref{B} and \eqref{D} we get 
	\begin{align}\label{F}
	xyf(x,y)
	-
	{\partial^2 f(x,y)\over\partial x\partial y}
	&=
	f(x,0)xy
	+
	{\partial f(0,y)\over \partial y} x
	\\[0,2cm]
	&\quad +
	\sum_{n=1}^{k-1} 
	\biggl(
	{\partial^n f(x,0)\over\partial y^n} \, x
	h_n(y)+{\partial^{n+1} f(0,y)\over\partial x^n\partial y}h_n(x)
	\biggr)  \nonumber
	\\[0,2cm]
	&\quad +
	{\partial^{k} f(x,b(y))\over\partial y^{k}} \, {xy^{k+1}\over k!} 
	-
	{\partial^k f(x,c(y))\over\partial y^k} \, {xy^{k-1}\over (k-1)!}  \nonumber
	\\[0,2cm]
	&\quad +
	{\partial^{k+1} f(\widetilde{b}(x),y)\over\partial x^{k}\partial y} \, {x^{k+1}\over k!}
	-
	{\partial^{k+1} f(\widetilde{c}(x),y)\over\partial x^k\partial y} \, {x^{k-1}\over (k-1)!}\,.
	\nonumber
	\end{align}
	Defining 
	\begin{align*}
	I_1=I_2&\coloneqq
	xyf(x,y)
	-
	{\partial f(x,y)\over\partial x\partial y}
	-
	f(x,0)xy
	-
	{\partial f(0,y)\over \partial y} x
	\\[0,2cm]
	&\quad -
	\sum_{n=1}^{k-1} 
	\biggl(
	{\partial^n f(x,0)\over\partial y^n} \, x
	h_n(y)+{\partial^{n+1} f(0,y)\over\partial x^n\partial y}h_n(x)
	\biggr)\,,
	\end{align*}
	we obtain from \eqref{E} and \eqref{F} the following two inequalities:
	\begin{align*}
	|I_1|
	&\leqslant
	2{\biggl\|{\partial^{k-1} f\over \partial y^{k-1}}\biggr\|_{\infty}}
	{|x||y|^k\over (k-1)!}
	+
	2{\biggl\|{\partial^{k} f\over \partial x^{k-1}\partial y}\biggr\|_{\infty}}
	{|x|^k\over (k-1)!}
	\\[0,2cm]
	&\quad +
	{\biggl\|{\partial^{k} f\over \partial y^{k}}\biggr\|_{\infty}}
	{|x||y|^{k-1}\over (k-1)!}
	+
	{\biggl\|{\partial^{k+1} f\over \partial x^{k}\partial y}\biggr\|_{\infty}}
	{|x|^{k-1}\over (k-1)!} 
	\end{align*}
	and
	\begin{align*}
	|I_2|\leqslant
	{\biggl\|{\partial^{k} f\over \partial y^{k}}\biggr\|_{\infty}}
	|x|
	\biggl({|y|^{k+1}\over k!}+{|y|^{k-1}\over (k-1)!} \biggr)
	+
	{\biggl\|{\partial^{k+1} f\over \partial x^{k}\partial y}\biggr\|_{\infty}}
	\biggl({|x|^{k+1}\over k!}+{|x|^{k-1}\over (k-1)!} \biggr)\,.
	\end{align*}
	Taking expectation of $|I_1|$ on 
	$D\coloneqq\{|x|\geqslant K_1, |y|\geqslant K_2\}$,
	with $K_1,K_2\geqslant 1$,
	for the first inequality, we get
	\begin{align}\label{eq1}
	\mathbb{E}(|I_1|:D)
	&\leqslant
	{2\over (k-1)!}\!
	\left(
	{\biggl\|{\partial^{k-1} f\over \partial y^{k-1}}\biggr\|_{\infty}}\!\!\!\!
	+
	{\biggl\|{\partial^{k} f\over \partial y^{k}}\biggr\|_{\infty}}
	\right)
	\mathbb{E}(|x|:|x|\geqslant K_1) \,
	\mathbb{E}(|y|^k:|y|\geqslant\! K_2)
	\\[0,2cm]
	&\quad
	+
	{2\over (k-1)!}
	\biggl(
	{\biggl\|{\partial^{k} f\over \partial x^{k-1}\partial y}\biggr\|_{\infty}}
	+
	{\biggl\|{\partial^{k+1} f\over \partial x^{k}\partial y}\biggr\|_{\infty}}
	\biggr)
	\mathbb{E}(|x|^k:|x|\geqslant K_2).   \nonumber
	\end{align}
	Now, by taking expectation of  $|I_2|$ on the set $\{|x|\leqslant K_1\}$
	we obtain
	\begin{align}\label{eq2}
	\mathbb{E}(|I_2|:|x|\leqslant K_1)
	&=
	\mathbb{E}(|I_2|:|x|\leqslant K_1, |y|\leqslant K_2)
	+
	\mathbb{E}(|I_2|:|x|\leqslant K_1, |y|\geqslant K_2)
	\\[0,2cm]
	&\leqslant
	{(k+1)K_1\over k!}
	\left(
	K_2
	{\biggl\|{\partial^{k} f\over \partial y^{k}}\biggr\|_{\infty}} 
	\mathbb{E}|y|^{k}
	+
	{\biggl\|{\partial^{k+1} f\over \partial x^{k}\partial y}\biggr\|_{\infty}}
	\mathbb{E}|x|^k
	\right) \nonumber
	\\[0,2cm]
	&\quad +
	{(k+1)K_1\over k!}
	\left(
	{\biggl\|{\partial^{k} f\over \partial y^{k}}\biggr\|_{\infty}}
	\mathbb{E}(|y|^{k+1}: |y|\geqslant K_2)
	+
	{\biggl\|{\partial^{k+1} f\over \partial x^{k}\partial y}\biggr\|_{\infty}}
	\mathbb{E}|x|^k
	\right) \nonumber
	\,.
	\end{align}
	Similarly, by taking expectation of  $|I_2|$ on the set $\{|y|\leqslant K_2\}$
	we see that 
	\begin{align}\label{eq3}
	&\hspace*{-0,5cm}
	\mathbb{E}(|I_2|:|y|\leqslant K_2)
	\leqslant
	{(k+1)K_1\over k!}
	\left(
	K_2
	{\biggl\|{\partial^{k} f\over \partial y^{k}}\biggr\|_{\infty}} 
	\mathbb{E}|y|^{k}
	+
	{\biggl\|{\partial^{k+1} f\over \partial x^{k}\partial y}\biggr\|_{\infty}}
	\mathbb{E}|x|^{k}
	\right)  \nonumber
	\\[0,2cm]
	\quad &\hspace*{-0,5cm}+
	{(k+1)\over k!}
	\left(
	K_2
	{\biggl\|{\partial^{k} f\over \partial y^{k}}\biggr\|_{\infty}} 
	\mathbb{E}(|x|: |x|\!\geqslant\! K_1) 
	\mathbb{E}|y|^k
	+
	{\biggl\|{\partial^{k+1} f\over \partial x^{k}\partial y}\biggr\|_{\infty}} 
	\mathbb{E}(|x|^{k+1}:|x|\!\geqslant K_1)
	\right)
	\,.
	\end{align}
	Since $x,y$ are two  random variables such
	that their first $k\geqslant 2$ moments match those of a Gaussian random variable,
	it follows that $\mathbb{E}h_n(x)=\mathbb{E}h_n(y)=0$,  $n=1,\ldots,k-1$. Then,
	\begin{align*}
	\biggl|\mathbb{E}xyf(x,y)-\mathbb{E}{\partial^2 f(x,y)\over\partial x\partial y}\biggr|
	&=
	|\mathbb{E}(I_1)|
	=
	|\mathbb{E}(I_2)|
	=
	|\mathbb{E}(I_1:D)+\mathbb{E}(I_2:D^c)|
	\\
	&\leqslant
	\mathbb{E}(|I_1|:D)
	\!+\!
	\mathbb{E}(|I_2|:|x|\leqslant K_1)
	\!+\!
	\mathbb{E}(|I_2|:|y|\leqslant K_2)\,.
	\end{align*}
	Finally, combining the above inequality 
	with \eqref{eq1}, \eqref{eq2} and \eqref{eq3}, we conclude the proof.
	
\end{proof}

%%%%%%%%%%%%%%%%%%%%%%%%%%%%%%%%%%%%%%%%%%%%%%%%%%%%%%%%%%%%%%%%%%%%%

\end{document}